\documentclass[
aip,jmp,
reprint,onecolumn,tightenlines,11pt,
a4paper,
longbibliography,
nofootinbib,
preprintnumbers,
citeautoscript
]{revtex4-1}
\usepackage[utf8]{inputenc}
\usepackage[T1]{fontenc}

\usepackage{amsfonts}
\usepackage{amsmath}
\usepackage{amsthm}
\usepackage{braket}
\usepackage{graphicx}
\usepackage{hyperref}
\usepackage{tabularx}

\hypersetup{citecolor=magenta}
\hypersetup{colorlinks=true}
\hypersetup{linkcolor=blue}
\hypersetup{urlcolor=blue}

\DeclareMathOperator{\Ci}{Ci}
\DeclareMathOperator{\CP}{Caterpillar}

\newcommand{\software}[1]{{\textsl{#1}}}
\newcommand{\FOR}{{\Xi}}
\newcommand{\OR}{{\xi}}
\newcommand{\abs}[1]{{\lvert{#1}\rvert}}

\newtheorem{theorem}{Theorem}
\newtheorem{lemma}[theorem]{Lemma}

\begin{document}

%%%%%%%%%%%%%%%%%%%%%%%%%%%%%%%%%%%%%%%%%%%%%%%%%%%%%%%%%%%%%%%%%%%

\title{Quantum state-independent contextuality requires 13 rays}

\author{Adán Cabello}
\email{adan@us.es}
\affiliation{%
Departamento de Física Aplicada II,
Universidad de Sevilla,
E-41012 Sevilla,
Spain}

\author{Matthias Kleinmann}
\email{matthias\_kleinmann001@ehu.eus}
\affiliation{%
Department of Theoretical Physics,
University of the Basque Country UPV/EHU,
P.O.\ Box 644,
E-48080 Bilbao,
Spain}

\author{José R.\ Portillo}
\email{josera@us.es}
\affiliation{%
Departamento de Matemática Aplicada I,
Universidad de Sevilla,
E-41012 Sevilla,
Spain}

%%%%%%%%%%%%%%%%%%%%%%%%%%%%%%%%%%%%%%%%%%%%%%%%%%%%%%%%%%%%%%%%%%%

\begin{abstract}
We show that, regardless of the dimension of the Hilbert space, there exists no 
 set of rays revealing state-independent contextuality with less than 13 rays.
This implies that the set proposed by Yu and Oh in dimension three [Phys.\ 
 Rev.\ Lett.\ \textbf{108}, 030402 (2012)] is actually the minimal set in 
 quantum theory.
This contrasts with the case of Kochen--Specker sets, where the smallest set 
 occurs in dimension four.
\end{abstract}

\maketitle

%%%%%%%%%%%%%%%%%%%%%%%%%%%%%%%%%%%%%%%%%%%%%%%%%%%%%%%%%%%%%%%%%%%

\section{Introduction}
Fifty years ago, Kochen and Specker \cite{KS67} answered the following 
 question:
Is it possible that, independently of which is the quantum state, the quantum 
 observables each possess a definite single value, regardless of whether they 
 are measured or not?
The Kochen--Specker (KS) theorem states that this is impossible if the 
 dimension of the underlying Hilbert space is larger than two.
One consequence of this theorem is the impossibility of reproducing quantum 
 theory in terms of noncontextual hidden variable theories, defined as those in 
 which the outcomes are independent of the context.
A context is a set of mutually compatible quantum observables.
In this sense, quantum theory is said to exhibit contextuality.

The original proof of the KS theorem had two other distinctive traits:
(i) It only used a finite set of observables with two outcomes, where one 
 outcome is represented by a rank-one projection onto a ray of the Hilbert 
 space.
Hereafter, as it is common in the literature, we will use ray as synonym of 
 self-adjoint rank-one projection.
(ii) The set is KS-uncolorable, i.e., it is impossible to assign values 1 or 0 
 to each ray while respecting that two orthogonal rays cannot both have 
 assigned 1, and 1 must be assigned to exactly one of $d$ mutually orthogonal 
 rays.
These restrictions are motivated by the observation that orthogonal rays 
 correspond to mutually exclusive outcomes of a sharp observable and $d$ 
 mutually orthogonal rays correspond to an exhaustive set of mutually exclusive 
 outcomes for a Hilbert space of dimension $d$.
KS-uncolorable sets of rays are called KS sets \cite{PMMM05}.

The original KS set had 117 rays in $d=3$, which can be grouped in 132 
 contexts.
There have been many efforts for finding simpler sets exhibiting 
 state-independent contextuality (SIC).
For instance, Peres and Mermin realized that, by considering observables not 
 represented by rank-one projections and replacing KS uncolorability by a 
 similar condition, one can find very compact sets of observables in $d=4$ and 
 $d=8$ \cite{Peres90, Mermin90}.
Still, these sets can be rewritten in terms of KS sets \cite{Peres91, KP95}.
So far, it has been shown \cite{PMMM05} that the KS set of minimum cardinality 
 occurs in $d=4$ and has 18 rays \cite{CEG96}.
It also has been proved \cite{PMMM05} that, in $d=3$, the KS set with minimum 
 cardinality has more than 22 and less than 32 rays \cite{CK95}.
On the other hand, the KS set with minimum number of contexts known occurs in 
 $d=6$ and has seven contexts (and 21 rays) \cite{LBPC14}.

A big step was the observation that SIC based on rays does not need to
 rely on KS-uncolorable sets.
It is enough that they lead to a state-independent violation of a 
 noncontextuality inequality.
This substantially simplifies the methods for revealing SIC in $d=3$.
Specifically, Yu and Oh singled out one set with 13 rays in $d=3$ \cite{YO12}.
The optimal state-independent noncontextuality inequalities for this set were 
 identified in Ref.~\onlinecite{KBLGC12}.
Sets of rays having a state-independent violation of a non-contextuality 
 inequality are called SIC sets.

Recent experiments testing SIC \cite{KZGKGCBR09, ARBC09, MRCL10, ZWDCLHYD12, 
 ZUZ13, DHANBSC13, CEGSXLC14, CAEGCXL14, JRO16} and an increasing number of 
 applications, such as device-independent secure communication \cite{HHHHPB10}, 
 local contextuality \cite{Cabello10, LHC16}, Bell inequalities revealing full 
 nonlocality \cite{AGACVMC12}, state-independent quantum dimension witnessing 
 \cite{GBCKL14}, and state-independent hardware certification \cite{CAEGCXL14}, 
 have stimulated the interest in the following question:
Which is the minimal set of rays needed for SIC?
It is known that, for $d=3$, the answer is 13 \cite{CKB15}, but it would be 
 well possible that the minimal set occurs in some higher dimension, as it 
 happens for KS sets.
Here we prove that this is not the case.

%%%%%%%%%%%%%%%%%%%%%%%%%%%%%%%%%%%%%%%%%%%%%%%%%%%%%%%%%%%%%%%%%%%

\section{Main result}
The basis of our proof is a condition identified by Ramanathan and Horodecki 
 \cite{RH14, CKB15} to be necessary for any SIC set in dimension $d$, namely 
 that the orthogonality graph $G$ of the set of rays has fractional chromatic 
 number $\chi_f(G)>d$.
The orthogonality graph of a SIC set is the graph in which orthogonal rays are 
 represented by adjacent vertices.
A coloring of $G$ is an assignment of colors to the vertices such that adjacent 
 vertices are associated with different colors.
$\chi_f(G)$ is the infimum of $\frac ab$ such that vertices have a set of $b$ 
 associated colors, out of $a$ colors, where adjacent vertices have associated 
 disjoint sets of colors.

Instead of considering all possible SIC sets of size $n$, we rather investigate 
 all graphs with $n$ vertices.
Then, we consider the nondegenerate orthogonal representations (ORs) of any 
 graph $G$.
An OR is an injection $\phi$, mapping the vertices of $G$ to rays, such that 
 adjacent vertices in $G$ are mapped to orthogonal rays.
The OR is faithful (FOR) if, conversely, any two orthogonal rays correspond to 
 an edge of $G$.
We denote by $\FOR(G)$ the smallest dimension of the Hilbert space which still 
 admits a FORs of $G$.
It then follows from the Ramanathan--Horodecki condition that $G$ is the 
 orthogonality graph of a SIC set only if $\chi_f(G)>\FOR(G)$.
Our main results is then as follows.

\begin{theorem}\label{Thm1}
Any graph $G$ with 12 or less vertices has $\chi_f(G)\le \FOR(G)$.
\end{theorem}

\noindent
Hence, according to quantum theory, no SIC set with less than 13 rays exists.

%%%%%%%%%%%%%%%%%%%%%%%%%%%%%%%%%%%%%%%%%%%%%%%%%%%%%%%%%%%%%%%%%%%

\section{Proof of Theorem~\ref{Thm1}}
We proceed by an exhaustive search for a counterexample, examining all 
 $166\,122\,463\,890$ nonisomorphic graphs with up to 12 vertices.
Applying a cascade of filters we eventually discard all graphs and prove this 
 way Theorem~\ref{Thm1}.
We start by introducing the criteria for defining these filters and then 
 explain our procedure providing intermediate results for each of the filters.

We denote by $V(G)$ and $E(G)$ the sets of vertices and edges of $G$, 
 respectively.
The complement $\overline G$ of $G$ is a graph that has the same vertices while 
the edges are the complemented set, i.e., $e\in E(\overline G)$ if and only if 
$e\notin E(G)$.
A subgraph $S$ of $G$ is any graph with $V(S)\subset V(G)$ and $E(S)\subset 
 E(G)$.
A subgraph is induced if $\overline S$ is also a subgraph of $\overline G$.
It is a simple observation that any (F)OR is also a (F)OR of any (induced) 
 subgraph.
Defining $\OR$
 analogously to $\FOR$, but for ORs,\footnote{%
The orthogonal rank of a graph is also sometimes denoted by $\OR$ 
 \cite{Cameron:2007EJC}, but there the minimum is taken without the restriction 
 that the OR is an injection.
This yields slightly different properties.}
 this proves the following.

\begin{lemma}\label{Lem2}
By definition, $\OR(G)\le \FOR(G)$.
If $S$ is a subgraph of $G$, then $\OR(S)\le \OR(G)$.
Similarly, if $S$ is an induced subgraph of $G$, then $\FOR(S)\le \FOR(G)$.
\end{lemma}

The union of two graphs $G_1\cup G_2$ consists of the disjoint union of the 
 respective vertex sets and edge sets.
The join $G_1+G_2$ of two graphs is the union of both graphs adding one edge 
 between any pair $(v_1, v_2)\in V(G_1)\times V(G_2)$.
The graph $K_1$ with one vertex and no edge takes a special role in the 
 following simple relations.

\begin{lemma}\label{Lem3}
For two graphs $G_1$ and $G_2$ and $f\in \set{\chi_f, \FOR, \OR}$, we have 
 $f(G_1\cup G_2)= \max[f(G_1), f(G_2)]$ and $f(G_1+ G_2)= f(G_1)+ f(G_2)$, with 
 the exceptions $\FOR(K_1\cup K_1)= 2$ and $\OR(K_1\cup K_1)= 2$.
\end{lemma}

\begin{proof}
For $\chi_f$ the relations are well-known, cf., e.g., 
 Ref.~\onlinecite{Scheinerman:1997}, Sec.~3.10.
For $\FOR$ and $\OR$ and the first relation, the maximum is at least a lower 
 bound, since any (F)OR of $G_1\cup G_2$ must also be a (F)OR of $G_1$ and of 
 $G_2$.
Conversely, if at least one of the graphs has more than one vertex then also 
 its (F)OR has at least dimension two.
This (F)OR can then be transformed by a unitary rotation, such that the image 
 of the (F)ORs of $G_1$ and $G_2$ are disjoint and also no rays are orthogonal.
Hence one can combine any two (F)ORs of $G_1$ and $G_2$ to a (F)OR in the 
 larger of the dimensions of both (F)ORs.
The second relation follows at once, noting that $\set{v_1,v_2} \in E(G_1+G_2)$ 
 if and only if either $v_1\in V(G_1)$ and $v_2\in V(G_2)$, or vice versa, or 
 $\set{v_1, v_2}\in E(G_1)$, or $\set{v_1, v_2}\in E(G_2)$.
Hence $\phi$ is a (F)OR for $G_1+G_2$ if and only if it is a (F)OR for $G_1$ 
 and $G_2$, and the spans of $\phi[V(G_1)]$ and $\phi[V(G_2)]$ are mutually 
 orthogonal.
\end{proof}

These relations are useful for our purposes since they imply that, if a graph 
 or its complement is not connected and $\chi_f(G)>\FOR(G)$, then this must 
 already be true for a subgraph of $G$.
Hence in our search we only need to consider connected graphs the complement of 
 whose are also connected.
Another important consequence of Lemma~\ref{Lem3} is that 
 $\OR(n\overline{K_2}+mK_1)= 2n+m$, where $K_\ell$ is the completely connected 
 graph with $\ell$ vertices \cite{Solis:2012, Solis:2015XXX}.
This implies $\FOR(G)\ge 2n+m$ as soon as $n\overline{K_2}+ mK_1$ is a subgraph 
 of $G$.
A weaker form of this condition is that if $K_\ell$ is a subgraph of $G$, then 
 $\FOR(G)\ge \ell$.

%%%%%%%%%%%%%%%%%%%%%%%%%%%%%%%%%%%%%%%%%%%%%%%%%%%%%%%%%%%%%%%%%%%
% Table I
%%%%%%%%%%%%%%%%%%%%%%%%%%%%%%%%%%%%%%%%%%%%%%%%%%%%%%%%%%%%%%%%%%%

\begin{table}
\begin{tabular}{cclrcr}
\hline
Graph name & In Fig.~\ref{Fig1} & \software{graph6} & $\FOR$ & Filter & 
Remaining\\\hline\hline
$\overline{H}$                &(a)& \texttt{Ebtw}      &5&(3.1)& $124\,220$\\
$\Ci_8(1,2)$                  &(b)& \texttt{Gbijmo}    &5&(3.2)& $124\,216$\\
$\overline{H}+K_1$            &---& \texttt{Fbvzw}     &6&(3.3)&   $4\,722$\\
$\overline{\CP_2^{3,2}}$      &(c)& \texttt{Fbtzw}     &6&(3.4)&        569\\
$\overline{\CP_3^{2,1,1}}$    &(d)& \texttt{Fbuzw}     &6&(3.5)&        400\\
$\Ci_{11}(1,2,3)\setminus\set{v}$ &(e)& \texttt{Ibgzmngjg} &6&(3.6)&    366\\
$\overline{H}+K_2$            &---& \texttt{Gzznnk}    &7&(3.7)&          0\\
\hline
\end{tabular}
\caption{\label{Tab1}%
List of graphs used for filtering via Lemma~\ref{Lem2}.
The graphs $\CP_k^{n_1, \dotsc, n_k}$ are linear graphs of length $k$, where 
 $n_v$ leafs are added to vertex $v$.
$H=\CP_2^{2,2}$, $\Ci_n(e_1, \dotsc, e_m)$ is the circulant graph, where each
 vertex is connected to its $e_1$th-, \ldots, $e_m$th-next neighbor.
$G\setminus\set{v}$ is the graph $G$ with one vertex removed.
Selected graphs are displayed in Fig.~\ref{Fig1}.
\software{graph6} is a standard graph data format widely used in
computer software \cite{graph6}.
The number $\FOR$ is the smallest dimension of any faithful nondegenerate
 orthogonal representation.
The last column shows the number of graphs remaining after filtering for the
 induced subgraph, cf.\ main text.
}
\end{table}

%%%%%%%%%%%%%%%%%%%%%%%%%%%%%%%%%%%%%%%%%%%%%%%%%%%%%%%%%%%%%%%%%%%
% Fig. 1
%%%%%%%%%%%%%%%%%%%%%%%%%%%%%%%%%%%%%%%%%%%%%%%%%%%%%%%%%%%%%%%%%%%

\begin{figure}
\includegraphics[width=.6\linewidth]{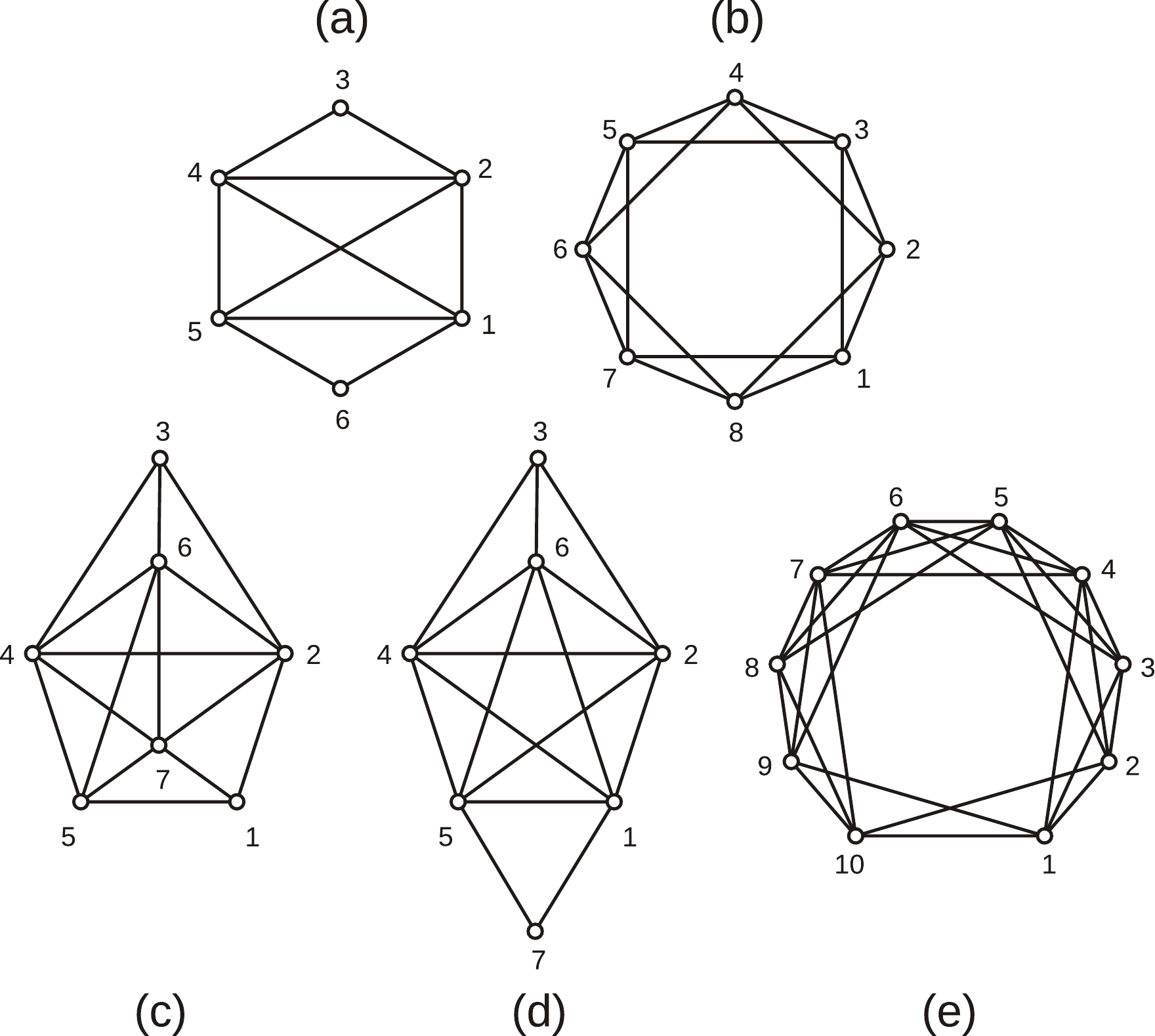}
\caption{\label{Fig1}%
Graphs from Table~\ref{Tab1}.
Graphs (a) and (b) have $\FOR=5$ and graphs (c)--(e) have $\FOR=6$.
The other two graphs from Table~\ref{Tab1} are obtained by adding one or two 
vertices to graph (a) each being connected to all other vertices.
}
\end{figure}

%%%%%%%%%%%%%%%%%%%%%%%%%%%%%%%%%%%%%%%%%%%%%%%%%%%%%%%%%%%%%%%%%%%

As a final ingredient to our proof, we use the seven graphs listed in
 Table~\ref{Tab1}.
If any of those graphs is an induced subgraph $S$ of $G$, then $\FOR(G)\ge 
 \FOR(S)$ applies.
The values of $\FOR(S)$ are obtained by construction, and due to 
 Lemma~\ref{Lem3} it is sufficient to study the five graphs in Fig.~\ref{Fig1}.
The construction is similar for all five graphs and we demonstrate the method 
 only for the most complicated case $\Ci_{11}(1,2,3)\setminus\set{v}$, cf.\ 
 Fig.~\ref{Fig1}~(e).
The vertices $\set{4,5,6,7}$ form the induced subgraph $K_4$ and, without loss 
 of generality, we can choose $\phi(4)= (1,0,0,0,0)$, $\phi(5)= (0,1,0,0,0)$, 
 $\phi(6)= (0,0,1,0,0)$, and $\phi(7)= (0,0,0,1,0)$.
Since vertex 3 is adjacent to the vertices $\set{4, 5, 6}$ and not adjacent to 
 vertex 7 or 8, and vertex 7 is adjacent to 8, we have $\phi(3)= (0,0,0,a,1)$ 
 with some $a\ne 0$.
By similar arguments, $\phi(2)= (0,0,b,-1/a^*,1)$ with $b\ne0$, and, by 
 symmetry, $\phi(8)= (c,0,0,0,1)$ and $\phi(9)= (-1/c^*,d,0,0,1)$, with $c, 
 d\ne 0$.
Using, that vertex 1 is adjacent to the vertices $\set{4,9,3,2}$, we have 
 $\phi(1)= (0,-1/d^*,-x/b^*,-1/a^*,1)$ with $x= 1+1/\abs{a}^2$, and, by 
 symmetry, $\phi(10)= (-1/c^*,-y/d^*,-1/b^*,0,1)$ with $y=1+1/\abs{c}^2$.
Eventually, vertex 1 and 10 are adjacent, implying $y/\abs{d}^2+x/\abs{b}^2+1= 
 0$, which is a contradiction.
However, it is straightforward to find a FOR in dimension 6, proving 
 $\FOR[\Ci_{11}(1,2,3)\setminus\set{v}]= 6$.

For all graphs with less than 13 vertices, we discard those graphs which 
 satisfy at least one of the following filter criteria:

\noindent
\begin{tabularx}{\linewidth}{@{}lX@{}}
(1)&$G$ or $\overline G$ is not connected.\\
(2.1)&$G$ has subgraph $K_\ell$, where $\chi_f(G)\le \ell$.\\
(2.2)&$G$ has subgraph $n\overline{K_2}+mK_1$, where $\chi_f(G)\le 2n+m<
 \chi_f(G)+1$ and $m\in\set{0,1}$.\\
(3.1)--(3.7)&$G$ has an induced subgraph $S$ from Table~\ref{Tab1} with 
 $\chi_f(G)\le \FOR(S)$.
\end{tabularx}

For obvious reasons, we fall back to a computer-based proof.
We use \software{geng} from the software package \software{nauty} 
 \cite{nauty14, nauty} to generate all nonisomorphic graphs.
The fractional chromatic number can be obtained by solving the linear program 
 \cite{Scheinerman:1997, SDT0},
\begin{equation}\begin{split}
\text{maximize:}& \sum_{v\in V(G)} x_v \\
\text{subject to:}& \sum_{v\in \mathcal I} x_v \le 1,
\text{ for all $\mathcal I$ of $G$}\\
& x_v\ge 0 \text{ for all $v\in V(G)$},
\end{split}\end{equation}
 where $\mathcal I$ are independent sets of $G$, i.e., sets of vertices where 
 all vertices are mutually nonadjacent.
We find optimal solutions to this program using the software package 
 \software{GLPK} \cite{GLPK} and verify the correctness of the solution by 
 applying the strong duality of linear programs, using an accuracy threshold of 
 $\epsilon=10^{-12}$.
We approximate the floating-point value obtained for $\chi_f$ by a rational 
 number with less than $\epsilon$ deviation, while constraining the denominator 
 to be not larger than $n m$, where $n$ is the number vertices of $G$ and $m$ 
 is the number of maximal independent sets.
This procedure always succeeds and ensures that the calculation of $\chi_f$ is 
 exact, despite floating-point arithmetic being used in intermediate steps.

%%%%%%%%%%%%%%%%%%%%%%%%%%%%%%%%%%%%%%%%%%%%%%%%%%%%%%%%%%%%%%%%%%%
% Table II
%%%%%%%%%%%%%%%%%%%%%%%%%%%%%%%%%%%%%%%%%%%%%%%%%%%%%%%%%%%%%%%%%%%

\begin{table}
\begin{tabular}{crrrrrr}
\hline
order&graphs&(1)&(2.1)&(2.2)\\\hline\hline
1&1&1&0&0\\
2&2&0&0&0\\
3&4&0&0&0\\
4&11&1&0&0\\
5&34&8&1&0\\
6&156&68&2&0\\
7&$1\,044$&662&28&0\\
8&$12\,346$&$9\,888$&456&0\\
9&$274\,668$&$247\,492$&$15\,954$&3\\
10&$12\,005\,168$&$11\,427\,974$&$957\,882$&98\\
11&$1\,018\,997\,864$&$994\,403\,266$&$99\,869\,691$&$5\,765$\\
12&$165\,091\,172\,592$&$163\,028\,488\,360$&$19\,715\,979\,447$&$560\,500$\\
\hline
\end{tabular}
\caption{\label{Tab2}%
Number of nonisomorphic graphs with 1--12 vertices.
(1)--(2.2) after filtering, cf.\ main text.
}
\end{table}

%%%%%%%%%%%%%%%%%%%%%%%%%%%%%%%%%%%%%%%%%%%%%%%%%%%%%%%%%%%%%%%%%%%

We apply all filters (1)--(3.7) consecutively so that each filter reduces the 
 number of candidate graphs.
The numbers of graphs remaining after each step are shown in Table~\ref{Tab2}, 
 for filters (1), (2.1), and (2.2), and as a function of the number of vertices 
 of the graph.
The list of $566\,366$ graphs remaining after filter (2.2) is available in 
 \software{graph6}-format \cite{data}.
For the filters (3.1)--(3.7), we show in Table~\ref{Tab1} the total number of 
 remaining graphs after each filter.
No graph remains after applying all filters, which proves Theorem~\ref{Thm1}.

%%%%%%%%%%%%%%%%%%%%%%%%%%%%%%%%%%%%%%%%%%%%%%%%%%%%%%%%%%%%%%%%%%%%%%

\section{Conclusions}
Contextuality is a fundamental feature of quantum observables and can be 
 completely detached from any features of the quantum state of the system.
This state-independent contextuality already occurs for the most elementary 
 case of observables being sharp and having only two outcomes, one of which is 
 nondegenerate; such observables can be represented by rays in a Hilbert space.
Here we have shown that state-independent contextuality with elementary 
 observables requires at least 13 different observables by performing an 
 exhaustive search over all cases with less observables.
The Yu--Oh set is an example of such 13 observables and is already realizable 
 on a three-level quantum system, which is the smallest quantum system allowing 
 for contextuality.
This is in contrast to the first instances of state-independent contextuality, 
 the Kochen--Specker sets, where the smallest set cannot be realized on a 
 three-level system.
Therefore, fifty years after the discovery of state-independent contextuality 
 in quantum theory, we finally have the answer to the question of which is the 
 simplest way to reveal it, i.e., which is the smallest set of elementary 
 observables exhibiting state-independent contextuality.

%%%%%%%%%%%%%%%%%%%%%%%%%%%%%%%%%%%%%%%%%%%%%%%%%%%%%%%%%%%%%%%%%%%

\begin{acknowledgments}
We thank the team of the Scientific Computing Center of Andalusia (CICA) for 
 their help with the distributed computing.
This work was supported by
 Project No.\ FIS2014-60843-P,
 ``Advanced Quantum Information'' (MINECO, Spain), with FEDER funds,
 the project ``Photonic Quantum Information'' (Knut and Alice Wallenberg 
 Foundation, Sweden),
 the EU (ERC Starting Grant GEDENTQOPT),
 and by the DFG (Forschungsstipendium KL 2726/2--1).
\end{acknowledgments}

%%%%%%%%%%%%%%%%%%%%%%%%%%%%%%%%%%%%%%%%%%%%%%%%%%%%%%%%%%%%%%%%%%%


\begin{thebibliography}{99}

\bibitem{KS67}
 S. Kochen and E. P. Specker,
 The problem of hidden variables in quantum mechanics,
 \href{http://dx.doi.org/10.1512/iumj.1968.17.17004}{J. Math. Mech. \textbf{17}, 59 (1967)}.

\bibitem{PMMM05}
 M. Pavičić, J.-P. Merlet, B. D. McKay, and N. D. Megill,
 Kochen--Specker vectors,
 \href{http://dx.doi.org/10.1088/0305-4470/38/7/013}{J. Phys. A: Math. Gen. \textbf{38}, 1577 (2005)}.

\bibitem{Peres90}
 A. Peres,
 Incompatible results of quantum measurements,
 \href{http://dx.doi.org/10.1016/0375-9601(90)90172-K}{Phys. Lett. A \textbf{151}, 107 (1990)}.

\bibitem{Mermin90}
 N. D. Mermin,
 Simple unified form for the major no-hidden-variables theorems,
 \href{http://dx.doi.org/10.1103/PhysRevLett.65.3373}{Phys. Rev. Lett. \textbf{65}, 3373 (1990)}.

\bibitem{Peres91}
 A. Peres,
 Two simple proofs of the Kochen--Specker theorem,
 \href{http://dx.doi.org/10.1088/0305-4470/24/4/003}{J. Phys. A: Math. Gen. \textbf{24}, L175 (1991)}.

\bibitem{KP95}
 M. Kernaghan and A. Peres,
 Kochen--Specker theorem for eight-dimensional space.
 \href{http://dx.doi.org/10.1016/0375-9601(95)00012-R}{Phys. Lett. A \textbf{198}, 1 (1995)}.

\bibitem{CEG96}
 A. Cabello, J. M. Estebaranz, and G. García-Alcaine,
 Bell--Kochen--Specker theorem: A proof with 18 vectors,
 \href{http://dx.doi.org/10.1016/0375-9601(96)00134-X}{Phys. Lett. A \textbf{212}, 183 (1996)}.

\bibitem{CK95}
 J. H. Conway and S. Kochen,
 reported in A. Peres,
 \emph{Quantum Theory: Concepts and Methods}
 (Kluwer, Dordrecht, 1995), p.~114.

\bibitem{LBPC14}
 P. Lisoněk, P. Badziąg, J. R. Portillo, and A. Cabello,
 Kochen--Specker set with seven contexts,
 \href{http://dx.doi.org/10.1103/PhysRevA.89.042101}{Phys. Rev. A \textbf{89}, 042101 (2014)}.

\bibitem{YO12}
 S. Yu and C. H. Oh,
 State-independent proof of Kochen--Specker theorem with 13 rays,
 \href{http://dx.doi.org/10.1103/PhysRevLett.108.030402}{Phys. Rev. Lett. \textbf{108}, 030402 (2012)}.

\bibitem{KBLGC12}
 M. Kleinmann, C. Budroni, J.-Å. Larsson, O. Gühne, and A. Cabello,
 Optimal inequalities for state-independent contextuality,
 \href{http://dx.doi.org/10.1103/PhysRevLett.109.250402}{Phys. Rev. Lett. \textbf{109}, 250402 (2012)}.

\bibitem{KZGKGCBR09}
 G. Kirchmair, F. Zähringer, R. Gerritsma, M. Kleinmann, O. Gühne, A. Cabello, R. Blatt, and C. F. Roos,
 State-independent experimental test of quantum contextuality,
 \href{http://dx.doi.org/10.1038/nature08172}{Nature (London) \textbf{460}, 494 (2009)}.

\bibitem{ARBC09}
 E. Amselem, M. Rådmark, M. Bourennane, and A. Cabello,
 State-independent quantum contextuality with single photons,
 \href{http://dx.doi.org/10.1103/PhysRevLett.103.160405}{Phys. Rev. Lett. \textbf{103}, 160405 (2009)}.

\bibitem{MRCL10}
 O. Moussa, C. A. Ryan, D. G. Cory, and R. Laflamme,
 Testing contextuality on quantum ensembles with one clean qubit,
 \href{http://dx.doi.org/10.1103/PhysRevLett.104.160501}{Phys. Rev. Lett. \textbf{104}, 160501 (2010)}.

\bibitem{ZWDCLHYD12}
 C. Zu, Y.-X. Wang, D.-L. Deng, X.-Y. Chang, K. Liu, P.-Y. Hou, H.-X. Yang, and L.-M. Duan,
 State-independent experimental test of quantum contextuality in an indivisible system,
 \href{http://dx.doi.org/10.1103/PhysRevLett.109.150401}{Phys. Rev. Lett. \textbf{109}, 150401 (2012)}.

\bibitem{ZUZ13}
 X. Zhang, M. Um, J. Zhang, S. An, Y. Wang, D.-L. Deng, C. Shen, L.-M. Duan, and K. Kim,
 State-independent experimental test of quantum contextuality with a single trapped ion,
 \href{http://dx.doi.org/10.1103/PhysRevLett.110.070401}{Phys. Rev. Lett. \textbf{110}, 070401 (2013)}.

\bibitem{DHANBSC13}
 V. D'Ambrosio, I. Herbauts, E. Amselem, E. Nagali, M. Bourennane, F. Sciarrino, and A. Cabello,
 Experimental implementation of a Kochen--Specker set of quantum tests,
 \href{http://dx.doi.org/10.1103/PhysRevX.3.011012}{Phys. Rev. X \textbf{3}, 011012 (2013)}.

\bibitem{CEGSXLC14}
 G. Cañas, S. Etcheverry, E. S. Gómez, C. Saavedra, G. B. Xavier, G. Lima, and A. Cabello,
 Experimental implementation of an eight-dimensional Kochen--Specker set and observation of its connection with the Greenberger--Horne--Zeilinger theorem,
 \href{http://dx.doi.org/10.1103/PhysRevA.90.012119}{Phys. Rev. A \textbf{90}, 012119 (2014)}.

\bibitem{CAEGCXL14}
 G. Cañas, M. Arias, S. Etcheverry, E. S. Gómez, A. Cabello, G. B. Xavier, and G. Lima,
 Applying the simplest Kochen--Specker set for quantum information processing,
 \href{http://dx.doi.org/10.1103/PhysRevLett.113.090404}{Phys. Rev. Lett. \textbf{113}, 090404 (2014)}.

\bibitem{JRO16}
 M. Jerger, Y. Reshitnyk, M. Oppliger, A. Potočnik, M. Mondal, A. Wallraff, K. Goodenough, S. Wehner, K. Juliusson, N. K. Langford, and A. Fedorov,
 Contextuality without nonlocality in a superconducting quantum system,
 \href{http://arxiv.org/abs/1602.00440}{\eprint{arXiv:1602.00440}}.

\bibitem{HHHHPB10}
 K. Horodecki, M. Horodecki, P. Horodecki, R. Horodecki, M. Pawłowski, and M. Bourennane,
 Contextuality offers device-independent security,
 \href{http://arxiv.org/abs/1006.0468}{\eprint{arXiv:1006.0468}}.

\bibitem{Cabello10}
 A. Cabello,
 Proposal for revealing quantum nonlocality via local contextuality,
 \href{http://dx.doi.org/10.1103/PhysRevLett.104.220401}{Phys. Rev. Lett. \textbf{104}, 220401 (2010)}.

\bibitem{LHC16}
 B.-H. Liu, X.-M. Hu, J.-S. Chen, Y.-F. Huang, Y.-J. Han, C.-F. Li, G.-C. Guo, and A. Cabello,
 Experimental test of the free will theorem,
 \href{http://arxiv.org/abs/1603.08254}{\eprint{arXiv:1603.08254}}.

\bibitem{AGACVMC12}
 L. Aolita, R. Gallego, A. Acín, A. Chiuri, G. Vallone, P. Mataloni, and A. Cabello,
 Fully nonlocal quantum correlations,
 \href{http://dx.doi.org/10.1103/PhysRevA.85.032107}{Phys. Rev. A \textbf{85}, 032107 (2012)}.

\bibitem{GBCKL14}
 O. Gühne, C. Budroni, A. Cabello, M. Kleinmann, and J.-Å. Larsson,
 Bounding the quantum dimension with contextuality,
 \href{http://dx.doi.org/10.1103/PhysRevA.89.062107}{Phys. Rev. A \textbf{89}, 062107 (2014)}.

\bibitem{CKB15}
 A. Cabello, M. Kleinmann, and C. Budroni,
 Necessary and sufficient condition for quantum state-independent contextuality,
 \href{http://dx.doi.org/10.1103/PhysRevLett.114.250402}{Phys. Rev. Lett. \textbf{114}, 250402 (2015)}.

\bibitem{RH14}
 R. Ramanathan and P. Horodecki,
 Necessary and sufficient condition for state-independent contextual measurement scenarios,
 \href{http://dx.doi.org/10.1103/PhysRevLett.112.040404}{Phys. Rev. Lett. \textbf{112}, 040404 (2014)}.

\bibitem{Cameron:2007EJC}
 P. J. Cameron, A. Montanaro, M. W. Newman, S. Severini, and A. Winter,
 On the quantum chromatic number of a graph,
 \href{http://www.combinatorics.org/Volume_14/Abstracts/v14i1r81.html}{Electr.  
 J. Comb. \textbf{14}, \#R81 (2007)}.

\bibitem{Scheinerman:1997}
 E. R. Scheinerman and D. H. Ullman,
 \href{http://www.ams.jhu.edu/ers/books/fractional-graph-theory-a-rational-approach-to-the-theory-of-graphs/}{\emph{Fractional Graph Theory. A Rational Approach to the Theory of Graphs}}
 (John Wiley \& Sons, New York, 1997).

\bibitem{Solis:2012}
 A. Solís,
 \emph{Algoritmos para la Resolución del Problema de Representación Ortogonal}
 (Master Thesis, Universidad de Sevilla, 2012)

\bibitem{Solis:2015XXX}
 A. Solís and J. R. Portillo,
 Orthogonal representation of graphs,
 \href{http://arxiv.org/abs/1504.03662}{\eprint{arXiv:1504.03662}}.

\bibitem{nauty14}
 B. D. McKay and A. Piperno,
 Practical graph isomorphism, II,
 \href{http://dx.doi.org/10.1016/j.jsc.2013.09.003}{J. Symb. Comput. \textbf{60}, 94 (2014)}.

\bibitem{nauty}
 nauty and Traces,
 \url{http://pallini.di.uniroma1.it/}.

\bibitem{SDT0}
 D. Bertsimas and J. N. Tsitsiklis,
 \emph{Introduction to Linear Optimization}
 (Athena Scientific, Belmont, Massachusetts, 1997).

\bibitem{GLPK}
 GNU Linear Programming Kit,
 \url{http://www.gnu.org/software/glpk/}.

\bibitem{data}
 \url{http://personal.us.es/josera/minSIC/},
 \software{sha256}-digest
\texttt{a23b}
\texttt{d030}
\texttt{d126}
\texttt{a3e5}
\texttt{44c0}
\texttt{f820}
\texttt{afcf}
\texttt{aa9a}
\texttt{ac31}
\texttt{a991}
\texttt{4ae1}
\texttt{416a}
\texttt{6a1a}
\texttt{682f}
\texttt{9bbe}
\texttt{2535}.

\bibitem{graph6}
 B. D. McKay,
 Description of graph6 and sparse6 encodings,
 \url{http://cs.anu.edu.au/~bdm/data/formats.txt}.

\end{thebibliography}
\end{document}